\documentclass[english]{amsart}
\usepackage{mathptmx}
\usepackage[T1]{fontenc}
\usepackage{amsthm}
\usepackage{graphicx}
\usepackage{amssymb}

\makeatletter
\numberwithin{equation}{section}
\numberwithin{figure}{section}
\theoremstyle{plain}
\newtheorem{thm}{Theorem}
 \theoremstyle{definition}
  \newtheorem{example}[thm]{Example}
  \theoremstyle{plain}
  \newtheorem{algorithm}[thm]{Algorithm}
  \theoremstyle{definition}
  \newtheorem{defn}[thm]{Definition}
  \theoremstyle{remark}
  \newtheorem{rem}[thm]{Remark}
  \theoremstyle{remark}
  \newtheorem{notation}[thm]{Notation}
  \theoremstyle{plain}
  \newtheorem{lem}[thm]{Lemma}

\def \Real {{\rm { l \kern -1.5pt R} }}

\makeatother

\usepackage{babel}

\begin{document}

\title{Divide-And-Conquer Computation of Cylindrical Algebraic Decomposition}

\author{Adam Strzebo\'nski}

\address{Wolfram Research Inc., 100 Trade Centre Drive, Champaign, IL 61820,
U.S.A. }

\email{adams@wolfram.com}

\date{May 3, 2013}
\begin{abstract}
We present a divide-and-conquer version of the Cylindrical Algebraic
Decomposition (CAD) algorithm. The algorithm represents the input
as a Boolean combination of subformulas, computes cylindrical algebraic
decompositions of solution sets of the subformulas, and combines the
results using the algorithm first introduced in \cite{S8}. We propose
a graph-based heuristic to find a suitable partitioning of the input
and present empirical comparison with direct CAD computation.
\end{abstract}
\maketitle

\section{Introduction}

A \emph{real polynomial system} in variables $x_{1},\ldots,x_{n}$
is a formula\[
S(x_{1},\ldots,x_{n})=\bigvee_{1\leq i\leq l}\bigwedge_{1\leq j\leq m}f_{i,j}(x_{1},\ldots,x_{n})\rho_{i,j}0\]
where $f_{i,j}\in\mathbb{R}[x_{1},\ldots,x_{n}]$, and each $\rho_{i,j}$
is one of $<,\leq,\geq,>,=,$ or $\neq$. 

A subset of $\mathbb{R}^{n}$ is \emph{semialgebraic} if it is a solution
set of a real polynomial system. 

Every semialgebraic set can be represented as a finite union of disjoint
\emph{cells} (see \cite{L}), defined recursively as follows.
\begin{itemize}
\item A cell in $\mathbb{R}$ is a point or an open interval.
\item A cell in $\mathbb{R}^{k+1}$ has one of the two forms\begin{eqnarray*}
 & \{(a_{1},\ldots,a_{k},a_{k+1}):(a_{1},\ldots,a_{k})\in C_{k}\wedge a_{k+1}=r(a_{1},\ldots,a_{k})\}\\
 & \{(a_{1},\ldots,a_{k},a_{k+1}):(a_{1},\ldots,a_{k})\in C_{k}\wedge r_{1}(a_{1},\ldots,a_{k})<a_{k+1}<r_{2}(a_{1},\ldots,a_{k})\}\end{eqnarray*}
where $C_{k}$ is a cell in $\mathbb{R}^{k}$, $r$ is a continuous
algebraic function, and $r_{1}$ and $r_{2}$ are continuous algebraic
functions, $-\infty$, or $\infty$, and $r_{1}<r_{2}$ on $C_{k}$. 
\end{itemize}
The Cylindrical Algebraic Decomposition (CAD) algorithm \cite{C,CJ,S7}
can be used to compute a cell decomposition of any semialgebraic set
presented by a real polynomial system. An alternative method of computing
cell decompositions is given in \cite{CMXY}. Cell decompositions
computed by the CAD algorithm can be represented directly \cite{B2,S7}
as cylindrical algebraic formulas (CAF; a precise definition is given
in Section \ref{CAFSec}). 
\begin{example}
The following formula $F(x,y,z)$ is a CAF representation of a cell
decomposition of the closed unit ball.\begin{eqnarray*}
F(x,y,z) & :=( & x=-1\wedge y=0\wedge z=0)\vee(-1<x<1\wedge b_{2}(x,y,z))\vee\\
 &  & (x=1\wedge y=0\wedge z=0)\\
b_{2}(x,y,z) & := & (y=R_{1}(x)\wedge z=0)\vee(R_{1}(x)<y<R_{2}(x)\wedge b_{2,2}(x,y,z))\vee\\
 &  & (y=R_{2}(x)\wedge z=0)\\
b_{2,2}(x,y,z) & := & z=R_{3}(x,y)\vee R_{3}(x,y)<z<R_{4}(x,y)\vee z=R_{4}(x,y)\end{eqnarray*}
where \begin{eqnarray*}
R_{1}(x) & = & Root_{y,1}(x^{2}+y^{2})=-\sqrt{1-x^{2}}\\
R_{2}(x) & = & Root_{y,2}(x^{2}+y^{2})=\sqrt{1-x^{2}}\\
R_{3}(x,y) & = & Root_{z,1}(x^{2}+y^{2}+z^{2})=-\sqrt{1-x^{2}-y^{2}}\\
R_{4}(x,y) & = & Root_{z,2}(x^{2}+y^{2}+z^{2})=\sqrt{1-x^{2}-y^{2}}\end{eqnarray*}

\end{example}
The CAF representation of a semialgebraic set $A$ can be used to
decide whether $A$ is nonempty, to find the minimal and maximal values
of the first coordinate of elements of $A$, to generate an arbitrary
element of $A$, to find a graphical representation of $A$, to compute
the volume of $A$ or to compute multidimensional integrals over $A$
(see \cite{S4}). 

In our ISSAC conference paper \cite{S8} we presented an algorithm,
\emph{CAFCombine}, computing Boolean operations on cylindrical algebraic
formulas. In this extended version of the paper we investigate how
\emph{CAFCombine} can be used to construct a divide-and-conquer algorithm
for computing a cylindrical algebraic decomposition. The divide-and-conquer
algorithm depends on the algorithm\emph{ Subdivide}. Given\emph{ }a
real polynomial system\emph{ $S(x_{1},\ldots,x_{n})$}, \emph{Subdivide}
finds a Boolean formula $\Phi$ and real polynomial systems $P_{1},\ldots,P_{m}$
such that \[
S(x_{1},\ldots,x_{n})\Leftrightarrow\Phi(P_{1}(x_{1},\ldots,x_{n}),\ldots,P_{m}(x_{1},\ldots,x_{n}))\]

\begin{algorithm}
(DivideAndConquerCAD)\\
Input: \emph{A real polynomial system} $S(x_{1},\ldots,x_{n})$.\textup{}\\
\textup{\emph{Output:}}\textup{ A cylindrical algebraic formula
$F(x_{1},\ldots,x_{n})$ equivalent to $S(x_{1},\ldots,x_{n})$.}
\begin{enumerate}
\item Use \emph{Subdivide} to find $\Phi$ and $P_{1},\ldots,P_{m}$ such
that \[
S\Leftrightarrow\Phi(P_{1},\ldots,P_{m})\]

\item If $m=1$ and $P_{1}=S$ then return $CAD(S)$.
\item For $1\leq i\leq m$, compute\[
F_{i}:=CAD(P_{i})\]

\item Use \emph{CAFCombine} to compute a CAF \textup{$F$ equivalent to\[
\Phi(F_{1},\ldots,F_{m})\]
}
\item Return \textup{$F$. }
\end{enumerate}
\end{algorithm}
The practical usefulness of \emph{DivideAndConquerCAD} depends on the
choice of the algorithm \emph{Subdivide}. Our experiments suggest
that \emph{DivideAndConquerCAD} is likely to be faster than a direct
CAD computation if $\Phi$ is a disjunction and, for any for $i\neq j$,
$P_{i}$ and $P_{j}$ contain few polynomials in common. We propose
an algorithm \emph{Subdivide} for polynomial systems $S$ given in
disjunctive normal form. The algorithm is based on the connectivity
structure of a graph whose vertices are the disjunction terms of $S$
and whose edges depend on polynomials shared between the disjunction
terms of $S$.
\begin{example}
Let $S(x,y)$ be the result of eliminating the quantifier from\begin{eqnarray*}
S_{0}(x,y) & :=\exists z & (z^{2}(-151x-740y-642)+z(-39x+285y-634)-241x-\\
 &  & 57y-985<0\wedge z^{2}(275x-144y+128)+z(94x-658y-267)+\\
 &  & 973x-810y+928=0\wedge z^{2}(-310x-224y+144)+\\
 &  & z(-256x-143y-77)+945x-260y+825\leq0)\end{eqnarray*}
using the virtual term substitution algorithm (\cite{W1}, \emph{Mathematica}
command $Resolve[S_{0},Reals]$). $S(x,y)$ is a disjunction of $43$
conjunctions of polynomial equations and inequalities.\\
\emph{Problem:} Find a cell decomposition for the solution set
of $S(x,y)$. \emph{}\\
\emph{Method 1}: A direct application of the CAD algorithm. The
computation takes $48$ seconds. \\
\emph{Method 2:} \emph{DivideAndConquerCAD} with \emph{Subdivide}
which represents $S(x,y)$ as a disjunction of $43$ subformulas,
each subformula equal to one of the conjunctions in $S(x,y)$. The
computation takes $8.5$ seconds.\\
\emph{Method 3}: \emph{DivideAndConquerCAD} with graph-based \emph{Subdivide}
which represents $S(x,y)$ as a disjunction of $5$ subformulas. The
subformulas are disjunctions of, respectively, $24$, $11,$ $6$,
$1$, and $1$ of the conjunctions in $S(x,y)$. The computation takes
$2.4$ seconds.
\end{example}
The paper is organized as follows. Section \ref{CAFSec} defines cylindrical
algebraic formulas. The algorithms \emph{CAFCombine} and \emph{Subdivide}
are presented in sections \ref{CAFCombineSec} and \ref{SubdivideSec}.
The last section contains experimental data comparing the performance
of \emph{DivideAndConquerCAD} and of direct CAD computation.

\section{\label{CAFSec}Cylindrical Algebraic Formulas}
\begin{defn}
A \emph{real algebraic function} given by \emph{defining polynomial}
$f\in\mathbb{\mathbb{Z}}[x_{1},\ldots,x_{n},y]$ and \emph{root number}
$p\in\mathbb{N}_{+}$ is the function\begin{equation}
Root_{y,p}f:\mathbb{R}^{n}\ni(x_{1},\ldots,x_{n})\longrightarrow Root_{y,p}f(x_{1},\ldots,x_{n})\in\mathbb{R}\label{rootfun}\end{equation}
where $Root_{y,p}f(x_{1},\ldots,x_{n})$ is the $p$-th real root
of $f$ treated as a univariate polynomial in $y$. The function is
defined for those values of  $x_{1},\ldots,x_{n}$ for which $f(x_{1},\ldots,x_{n},y)$
has at least $p$ real roots. The real roots are ordered by the increasing
value, counting multiplicities. A real algebraic number $Root_{y,p}f\in\mathbb{R}$
given by \emph{defining polynomial} $f\in\mathbb{Z}[y]$ and \emph{root
number} $p$ is the $p$-th real root of $f$. Let $Alg$ be the set
of real algebraic numbers and for $C\subseteq\mathbb{R}^{n}$ let
$Alg_{C}$ denote the set of all algebraic functions defined and continuous
on $C$. (See \cite{S2,S4} for more details on how algebraic numbers
and functions can be implemented in a computer algebra system.)
\end{defn}

\begin{defn}
A set $P\subseteq\mathbb{R}[x_{1},\ldots,x_{n},y]$ is \emph{delineable}
over $C\subseteq\mathbb{R}^{n}$ iff 
\begin{enumerate}
\item $\forall f\in P\:\exists k_{f}\in\mathbb{N\:}\forall a\in C\;\sharp\{b\in\mathbb{R}:f(a,b)=0\}=k_{f}$.
\item For any $f\in P$ and $1\leq p\leq k_{f}$, $Root_{y,p}f$ is a continuous
function on $C$.
\item \begin{eqnarray*}
\forall f,g\in P &  & (\exists a\in C\: Root_{y,p}f(a)=Root_{y,q}g(a)\Leftrightarrow\\
 &  & \forall a\in C\: Root_{y,p}f(a)=Root_{y,q}g(a))\end{eqnarray*}

\end{enumerate}
\end{defn}

\begin{defn}
\label{CSAC}A \emph{cylindrical system of algebraic constraints}
in variables $x_{1},\ldots,x_{n}$ is a sequence $A=(A_{1},\ldots,A_{n})$
satisfying the following conditions. 
\begin{enumerate}
\item For $1\leq k\leq n$, $A_{k}$ is a set of formulas\begin{eqnarray*}
A_{k} & = & \{a_{i_{1},\ldots,i_{k}}(x_{1},\ldots,x_{k}):1\leq i_{1}\leq m\wedge1\leq i_{2}\leq m_{i_{1}}\wedge\ldots\wedge1\leq i_{k}\leq m_{i_{1},\ldots,i_{k-1}}\}\end{eqnarray*}
 
\item For each $1\leq i_{1}\leq m$, $a_{i_{1}}(x_{1})$ is $true$ or \[
x_{1}=r\]
 where $r\in Alg$, or \[
r_{1}<x_{1}<r_{2}\]
where $r_{1}\in Alg\cup\{-\infty\}$, $r_{2}\in Alg\cup\{\infty\}$
and $r_{1}<r_{2}$. Moreover, if $s_{1},s_{2}\in Alg\cup\{-\infty,\infty\}$,
$s_{1}$ appears in $a_{u}(x_{1})$, $s_{2}$ appears in $a_{v}(x_{1})$
and $u<v$ then $s_{1}\leq s_{2}$.
\item Let $k<n$, $I=(i_{1},\ldots,i_{k})$ and let $C_{I}\subseteq\mathbb{R}^{k}$
be the solution set of\begin{equation}
a_{i_{1}}(x_{1})\wedge a_{i_{1},i_{2}}(x_{1},x_{2})\wedge\ldots\wedge a_{i_{1},\ldots,i_{k}}(x_{1},\ldots,x_{k})\label{supp}\end{equation}

\begin{enumerate}
\item For each $1\leq i_{k+1}\leq m_{I}$, \[
a_{i_{1},\ldots,i_{k},i_{k+1}}(x_{1},\ldots,x_{k},x_{k+1})\]
 is $true$ or\begin{equation}
x_{k+1}=r(x_{1},\ldots,x_{k})\label{eqbd}\end{equation}
and $r\in Alg_{C_{I}}$, or\begin{equation}
r_{1}(x_{1},\ldots,x_{k})<x_{k+1}<r_{2}(x_{1},\ldots,x_{k})\label{ineqbd}\end{equation}
where $r_{1}\in Alg_{C_{I}}\cup\{-\infty\}$, $r_{2}\in Alg_{C_{I}}\cup\{\infty\}$
and $r_{1}<r_{2}$ on $C_{I}$. 
\item If $s_{1},s_{2}\in Alg_{C_{I}}\cup\{-\infty,\infty\}$, $s_{1}$ appears
in \[
a_{i_{1},\ldots,i_{k},u}(x_{1})\]
$s_{2}$ appears in \[
a_{i_{1},\ldots,i_{k},v}(x_{1})\]
 and $u<v$ then $s_{1}\leq s_{2}$ on $C_{I}$.
\item Let $P_{I}\subseteq\mathbb{Z}[x_{1},\ldots,x_{k},x_{k+1}]$ be the
set of defining polynomials of all real algebraic functions that appear
in formulas $a_{J}$ for $J=(i_{1},\ldots,i_{k},i_{k+1})$, $1\leq i_{k+1}\leq m_{I}$.
Then $P_{I}$ is\emph{ }delineable\emph{ }over $C_{I}$.
\end{enumerate}
\end{enumerate}
\end{defn}

\begin{defn}
\label{CAF}Let $A$ be a cylindrical system of algebraic constraints
in variables $x_{1},\ldots,x_{n}$. Define

\begin{eqnarray*}
b_{i_{1},\ldots,i_{n}}(x_{1},\ldots,x_{n}) & := & true\end{eqnarray*}
For $2\leq k\leq n$, \emph{level $k$ cylindrical algebraic subformulas}
given by $A$ are the formulas\begin{eqnarray*}
 & b_{i_{1},\ldots,i_{k-1}}(x_{1},\ldots,x_{n}):=\bigvee_{1\leq i_{k}\leq m_{i_{1},\ldots,i_{k-1}}}a_{i_{1},\ldots,i_{k}}(x_{1},\ldots,x_{k})\wedge b_{i_{1},\ldots,i_{k}}(x_{1},\ldots,x_{n})\end{eqnarray*}
The \emph{support cell} of $b_{i_{1},\ldots,i_{k-1}}$ is the solution
set \[
C_{i_{1},\ldots,i_{k-1}}\subseteq\mathbb{R}^{k}\]
 of \[
a_{i_{1}}(x_{1})\wedge a_{i_{1},i_{2}}(x_{1},x_{2})\wedge\ldots\wedge a_{i_{1},\ldots,i_{k-1}}(x_{1},\ldots,x_{k-1})\]
\emph{The cylindrical algebraic formula (CAF)} given by $A$ is the
formula\[
F(x_{1},\ldots,x_{n}):=\bigvee_{1\leq i_{1}\leq m}a_{i_{1}}(x_{1})\wedge b_{i_{1}}(x_{1},\ldots,x_{n})\]
\end{defn}
\begin{rem}
Let $F(x_{1},\ldots,x_{n})$ be a CAF given by a cylindrical system
of algebraic constraints $A$. Then
\begin{enumerate}
\item For $1\leq k\leq n$, sets $C_{i_{1},\ldots,i_{k}}$ are cells in
$\mathbb{R}^{k}$.
\item Cells\begin{eqnarray*}
 & \{C_{i_{1},\ldots,i_{n}}:1\leq i_{1}\leq m\wedge1\leq i_{2}\leq m_{i_{1}}\wedge\ldots\wedge1\leq i_{n}\leq m_{i_{1},\ldots,i_{n-1}}\}\end{eqnarray*}
form a decomposition of the solution set $S_{F}$ of $F$, i.e. they
are disjoint and their union is equal to $S_{F}$.
\end{enumerate}
\end{rem}
\begin{proof}
Both parts of the remark follow from the definitions of $A$ and $F$.\end{proof}
\begin{rem}
Given a real polynomial system $S(x_{1},\ldots,x_{n})$ a version
of the CAD algorithm can be used to find a CAF $F(x_{1},\ldots,x_{n})$
equivalent to $S(x_{1},\ldots,x_{n})$.\end{rem}
\begin{proof}
The version of CAD described in \cite{S7} returns a CAF equivalent
to the input system.
\end{proof}

\section{\label{CAFCombineSec}Algorithm CAFCombine}

In this section we describe the algorithm \emph{CAFCombine}. The algorithm
is a modified version of the CAD algorithm. We describe only the modification.
For details of the CAD algorithm see \cite{CJ,C}. Our implementation
is based on the version of CAD described in \cite{S7}.
\begin{defn}
Let $P\subseteq\mathbb{R}[x_{1},\ldots,x_{n}]$ be a finite set of
polynomials and let $\overline{P}$ be the set of irreducible factors
of elements of $P$. $W=(W_{1},\ldots,W_{n})$ is a \emph{projection
sequence} for $P$ iff
\begin{enumerate}
\item \emph{Projection sets} $W_{1},\ldots,W_{n}$ are finite sets of irreducible
polynomials.
\item For $1\leq k\leq n$, $\overline{P}\cap(\mathbb{R}[x_{1},\ldots,x_{k}]\setminus\mathbb{R}[x_{1},\ldots,x_{k-1}])\subseteq W_{k}\subseteq\mathbb{R}[x_{1},\ldots,x_{k}]\setminus\mathbb{R}[x_{1},\ldots,x_{k-1}]$.
\item If $k<n$ and all polynomials of $W_{k}$ have constant signs on a
cell $C\subseteq\mathbb{R}^{k}$, then all polynomials of $W_{k+1}$
that are not identically zero on $C\times\mathbb{R}$ are delineable
over $C$. 
\end{enumerate}
\end{defn}
\begin{rem}
For an arbitrary finite set $P\subseteq\mathbb{R}[x_{1},\ldots,x_{n}]$
a projection sequence can be computed using Hong's projection operator
\cite{H}. McCallum's projection operator \cite{MC1,MC2} gives smaller
projection sets for well-oriented sets $P$.

If $P\subseteq Q\subseteq\mathbb{R}[x_{1},\ldots,x_{n}]$ and $W$
is a projection sequence for $Q$ then $W$ is a projection sequence
for $P$.\end{rem}
\begin{notation}
For a CAF $F$, let $P_{F}$ denote the set of defining polynomials
of all algebraic numbers and functions that appear in $F$ . 
\end{notation}
First let us prove the following rather technical effective lemmas
that will be used in the algorithm. We use notation of Definition
\ref{CAF}.
\begin{lem}
\label{lem1}Let \[
F(x_{1},\ldots,x_{n}):=\bigvee_{1\leq i_{1}\leq m}a_{i_{1}}(x_{1})\wedge b_{i_{1}}(x_{1},\ldots,x_{n})\]
be a CAF and let $-\infty=r_{0}<r_{1}<\ldots<r_{l}<r_{l+1}=\infty$
be such that all real roots of elements of $P_{F}\cap\mathbb{R}[x_{1}]$
are among $r_{1},\ldots,r_{l}$. Let $a(x_{1})$ be either $x_{1}=r_{j}$
for some $1\leq j\leq l$, or $r_{j}<x_{1}<r_{j+1}$ for some $0\leq j\leq l$,
and let $C_{a}$ be the solution set of $a$. Then \begin{equation}
\forall x_{1}\in C_{a}\: F(x_{1},\ldots,x_{n})\Leftrightarrow G(x_{1},\ldots,x_{n})\label{eqv01}\end{equation}
and one of the following two statements is true
\begin{enumerate}
\item There exist $1\leq i_{1}\leq m$ such that $C_{a}\subseteq C_{i_{1}}$
and $G(x_{1},\ldots,x_{n})=b_{i_{1}}(x_{1},\ldots,x_{n})$.
\item For all $1\leq i_{1}\leq m$, $C_{a}\cap C_{i_{1}}=\emptyset$ and
$G(x_{1},\ldots,x_{n})=false$ 
\end{enumerate}
Moreover, given $F$ and $a$, $G$ can be found algorithmically.\end{lem}
\begin{proof}
Let $r$ be an algebraic number that appears in $a_{i_{1}}$. Then
$r=Root_{x_{1},p}f$ for some $f\in P_{F}$. Hence, $r=r_{j_{0}}$
for some $1\leq j_{0}\leq l$ and the value of $j_{0}$ can be determined
algorithmically. If $a$ is $x_{1}=r_{j}$, then $C_{a}\cap C_{i_{1}}\neq\emptyset$
iff $a_{i_{1}}$ is either $x_{1}=r_{j}$ or $r_{u}<x_{1}<r_{v}$
with $u<j<v$. In both cases $C_{a}\subseteq C_{i_{1}}$. If $a$
is $r_{j}<x_{k}<r_{j+1}$, then $C_{a}\cap C_{i_{1}}\neq\emptyset$
iff $a_{i_{1}}$ is $r_{u}<x_{1}<r_{v}$ with $u\leq j$ and $v\geq j+1$.
In this case also $C_{a}\subseteq C_{i_{1}}$. Equivalence (\ref{eqv01})
follows from the statements (1) and (2).\end{proof}
\begin{lem}
\label{lem2}Let $2\leq k\leq n$, let\begin{eqnarray*}
b_{i_{1},\ldots,i_{k-1}}(x_{1},\ldots,x_{n}) & := & \bigvee_{1\leq i_{k}\leq m_{i_{1},\ldots,i_{k-1}}}a_{i_{1},\ldots,i_{k}}(x_{1},\ldots,x_{k})\wedge b_{i_{1},\ldots,i_{k}}(x_{1},\ldots,x_{n})\end{eqnarray*}
be a level \emph{$k$ }cylindrical algebraic subformula of a CAF $F$,
let \emph{$W=(W_{1},\ldots,W_{n})$ }be a projection sequence for
$P_{F}$. Let $C\subseteq\mathbb{R}^{k-1}$ be a cell such that all
polynomials of $W_{k-1}$ have constant signs on $C$ and \textup{$C\subseteq C_{i_{1},\ldots,i_{k-1}}$.}
Let $(c_{1},\ldots,c_{k-1})\in C$ and let $d_{1}<\ldots<d_{l}$ be
all real roots of $\{f(c_{1},\ldots,c_{k-1},x_{k}):f\in W_{k}\}$.
For $1\leq j\leq l$, let $r_{j}:=Root_{x_{k},p}f$, where $f\in W_{k}$
and $d_{j}$ is the $p$-th root of $f(c_{1},\ldots,c_{k-1},x_{k})$.
Let $a(x_{1},\ldots,x_{k})$ be either $x_{k}=r_{j}$ for some $1\leq j\leq l$,
or $r_{j}<x_{k}<r_{j+1}$ for some $0\leq j\leq l$, where $r_{0}:=-\infty$
and $r_{l+1}:=\infty$ and let\[
C_{a}:=\{(x_{1},\ldots,x_{k}):(x_{1},\ldots,x_{k-1})\in C\wedge a(x_{1},\ldots,x_{k})\}\]
Then \begin{equation}
\forall(x_{1},\ldots,x_{k})\in C_{a}\: b_{i_{1},\ldots,i_{k-1}}(x_{1},\ldots,x_{n})\Leftrightarrow G(x_{1},\ldots,x_{n})\label{eqv1}\end{equation}
and one of the following two statements is true
\begin{enumerate}
\item There exist $1\leq i_{k}\leq m_{i_{1},\ldots,i_{k-1}}$ such that
\[
C_{a}\subseteq C_{i_{1},\ldots,i_{k-1},i_{k}}\]
 and \[
G(x_{1},\ldots,x_{n})=b_{i_{1},\ldots,i_{k}}(x_{1},\ldots,x_{n})\]

\item For all $1\leq i_{k}\leq m_{i_{1},\ldots,i_{k-1}}$ \[
C_{a}\cap C_{i_{1},\ldots,i_{k-1},i_{k}}=\emptyset\]
 and \[
G(x_{1},\ldots,x_{n})=false\]

\end{enumerate}
Moreover, given $b_{i_{1},\ldots,i_{k-1}}$, $a$, $(c_{1},\ldots,c_{k-1})$,
$d_{1},\ldots,d_{l}$ and the multiplicity of $d_{j}$ as a root of
$f$, for all $1\leq j\leq l$ and $f\in W_{k}$, $G$ can be found
algorithmically.\end{lem}
\begin{proof}
Let $r$ be an algebraic function that appears in $a_{i_{1},\ldots,i_{k}}$.
Then $r=Root_{x_{k},p}f$ for some $f\in P_{F}$. By Definition \ref{CSAC},
$r$ is defined and continuous on $C$. Since $W$ is a projection
sequence for $P_{F}$, all factors of $f$ that depend on $x_{k}$
are elements of $W_{k}$. Hence, $r(c_{1},\ldots,c_{k-1})=d_{j_{0}}$
for some $1\leq j_{0}\leq l$. Since $d_{j_{0}}$ is the $p$-th of
real roots of factors of $f$, multiplicities counted, if the multiplicity
of $d_{j}$ as a root of $f$ is known for all $1\leq j\leq l$ and
$f\in W_{k}$, the value of $j_{0}$ can be determined algorithmically.
Since all polynomials of $W_{k-1}$ have constant signs on $C$, all
elements of $W_{k}$ that are not identically zero on $C$ are delineable
over $C$. Therefore, $r=r_{j_{0}}$ and $r_{1}<\ldots<r_{l}$ on
$C$. If $a$ is $x_{k}=r_{j}$, then $C_{a}\cap C_{i_{1},\ldots,i_{k-1},i_{k}}\neq\emptyset$
iff $a_{i_{1},\ldots,i_{k}}$ is either $x_{k}=r_{j}$ or $r_{u}<x_{k}<r_{v}$
with $u<j<v$. In both cases $C_{a}\subseteq C_{i_{1},\ldots,i_{k-1},i_{k}}$.
If $a$ is $r_{j}<x_{k}<r_{j+1}$, then $C_{a}\cap C_{i_{1},\ldots,i_{k-1},i_{k}}\neq\emptyset$
iff $a_{i_{1},\ldots,i_{k}}$ is $r_{u}<x_{k}<r_{v}$ with $u\leq j$
and $v\geq j+1$. In this case also $C_{a}\subseteq C_{i_{1},\ldots,i_{k-1},i_{k}}$.
Equivalence (\ref{eqv1}) follows from the statements (1) and (2).
\end{proof}
Let us now describe two subalgorithms used in \emph{CAFCombine}.
The first, recursive, subalgorithm requires its input to satisfy the
following conditions.
\begin{enumerate}
\item \emph{$W=(W_{1},\ldots,W_{n})$ }is a projection sequence for $P_{F_{1}}\cup\ldots\cup P_{F_{m}}$,
where \emph{\[
F_{1}(x_{1},\ldots,x_{n}),\ldots,F_{m}(x_{1},\ldots,x_{n})\]
} are cylindrical algebraic formulas.
\item $(c_{1},\ldots,c_{k-1})\in C$, \emph{$2\leq k\leq n$} and $C\subseteq\mathbb{R}^{k-1}$
is a cell such that all polynomials of $W_{k-1}$ have constant signs
on $C$.
\item Each $B_{j}$ is a\emph{ }level $k$ cylindrical algebraic formula\emph{
}of $F_{j}$ or $false$.
\item $C$ is contained in the intersection of support cells of all $B_{j}$
that are not $false$.
\item $\Phi(p_{1},\ldots,p_{m})$ is a Boolean formula.\end{enumerate}
\begin{algorithm}
(Lift)\\
Input: \emph{$(c_{1},\ldots,c_{k-1})\in\mathbb{R}^{k-1}$, $W$,}\textup{\emph{
$B_{1},\ldots,B_{m}$}}\textup{,}\textup{\emph{ $\Phi$.}}\\
\textup{\emph{Output:}}\textup{ A }\emph{level $k$ cylindrical
algebraic subformula}\textup{ \[
b(x_{1},\ldots,x_{n})\]
 such that\begin{equation}
\forall(x_{1},\ldots,x_{k-1})\in C\: b(x_{1},\ldots,x_{n})\Leftrightarrow\Phi(B_{1}(x_{1},\ldots,x_{n}),\ldots,B_{m}(x_{1},\ldots,x_{n}))\label{equiv2}\end{equation}
}
\begin{enumerate}
\item Let $d_{1}<\ldots<d_{l}$ be all real roots of \[
\{f(c_{1},\ldots,c_{k-1},x_{k}):f\in W_{k}\}\]

\item For $1\leq i\leq l$, let $r_{i}:=Root_{x_{k},p}f$, where $f\in W_{k}$
and $d_{i}$ is the $p$-th root of $f(c_{1},\ldots,c_{k-1},x_{k})$.
\item For $f\in W_{k}$, if $d_{i}$ is a root of $f(c_{1},\ldots,c_{k-1},x_{k})$,
let $M(f,i)$ be its multiplicity, otherwise $M(f,i):=0$.
\item For $1\leq i\leq l$, set $a_{2i}(x_{1},\ldots,x_{k}):=(x_{k}=r_{i})$
and $c_{k,2i}:=d_{i}$. 
\item For $0\leq i\leq l,$ set \[
a_{2i+1}(x_{1},\ldots,x_{k}):=(r_{i}<x_{k}<r_{i+1})\]
where $r_{0}:=-\infty$ and $r_{l+1}:=\infty$, and pick $c_{k,2i+1}\in(d_{i},d_{i+1})\cap\mathbb{Q}$,
where $d_{0}:=-\infty$ and $d_{l+1}:=\infty$.
\item For $1\leq i\leq2l+1$

\begin{enumerate}
\item For $1\leq j\leq m$, if $B_{j}=false$, set $G_{j}=false$, otherwise
let $G_{j}$ be the formula $G$ found using Lemma \ref{lem2} applied
to $B_{j}$, $a_{i}$, $(c_{1},\ldots,c_{k-1})$, \\
$d_{1},\ldots,d_{l}$ and $M$.
\item Let $\Psi:=\Phi(G_{1},\ldots,G_{m})$. If $\Psi$ is $true$ or $false$,
set $b_{i}(x_{1},\ldots,x_{n}):=\Psi$.
\item Otherwise set $b_{i}(x_{1},\ldots,x_{n})$ to\[
Lift((c_{1},\ldots,c_{k-1},c_{k,i}),W,G_{1},\ldots,G_{m},\Phi)\]

\end{enumerate}
\item Return \[
b(x_{1},\ldots,x_{n}):=\bigvee_{1\leq i\leq2l+1}a_{i}(x_{1},\ldots,x_{k})\wedge b_{i}(x_{1},\ldots,x_{n})\]

\end{enumerate}
\end{algorithm}
The second subalgorithm requires its input to satisfy the following
conditions.
\begin{enumerate}
\item Either $a(x_{1})=(x_{1}=r)$, where $r\in Alg$, or $a(x_{1})=(r<x_{1}<s)$,
where $r\in Alg\cup\{-\infty\}$ and $s\in Alg\cup\{\infty\}$.
\item For $1\leq j\leq m$, $B_{j}$ is a\emph{ }level $2$ cylindrical
algebraic subformula\emph{ }of a CAF\[
F_{j}(x_{1},\ldots,x_{n}):=a(x_{1})\wedge B_{j}(x_{1},\ldots,x_{n})\]

\item $\Phi(p_{1},\ldots,p_{m})$ is a Boolean formula.\end{enumerate}
\begin{algorithm}
(CombineStacks)\\
Input:\emph{ $a$,}\textup{\emph{ $B_{1},\ldots,B_{m}$}}\textup{,}\textup{\emph{
$\Phi$.}}\textup{}\\
\textup{\emph{Output:}}\textup{ A CAF $F(x_{1},\ldots,x_{n})$
such that\begin{equation}
F(x_{1},\ldots,x_{n})\Longleftrightarrow\Phi(F_{1}(x_{1},\ldots,x_{n}),\ldots,F_{m}(x_{1},\ldots,x_{n}))\label{equiv1}\end{equation}
}
\begin{enumerate}
\item Let $W:=(W_{1},\ldots,W_{n})$ be a projection sequence for $P_{F_{1}}\cup\ldots\cup P_{F_{m}}$.
\item If $a(x_{1})=(x_{1}=r)$ then set $l:=0$, $a_{1}(x_{1}):=(x_{1}=r)$,
and $c_{1,1}:=r$. 
\item If $a(x_{1})=(r<x_{1}<s)$ then

\begin{enumerate}
\item Let $r_{1}<\ldots<r_{l}$ be all real roots of elements of $W_{1}$
in $(r,s)$. 
\item For $1\leq i\leq l,$ set $a_{2i}(x_{1}):=(x_{1}=r_{i})$ and $c_{1,2i}:=r_{i}$. 
\item For $0\leq i\leq l,$ set $a_{2i+1}(x_{1}):=(r_{i}<x_{1}<r_{i+1})$
and pick $c_{1,2i+1}\in(r_{i},r_{i+1})\cap\mathbb{Q}$, where $r_{0}:=r$
and $r_{l+1}:=s$.
\end{enumerate}
\item For $1\leq i\leq2l+1$ set\[
b_{i}(x_{1},\ldots,x_{n}):=Lift((c_{1,i}),W,B_{1},\ldots,B_{m},\Phi)\]

\item Return \[
F(x_{1},\ldots,x_{n}):=\bigvee_{1\leq i\leq2l+1}a_{i}(x_{1})\wedge b_{i}(x_{1},\ldots,x_{n})\]

\end{enumerate}
\end{algorithm}
We can now describe the algorithm \emph{CADCombine} (cf. \cite{S8},
Algorithm 17).
\begin{algorithm}
(CAFCombine)\\
Input:\emph{ Cylindrical algebraic formulas}\textup{\emph{ \[
F_{1}(x_{1},\ldots,x_{n}),\ldots,F_{m}(x_{1},\ldots,x_{n})\]
}}\textup{ and a Boolean formula}\textup{\emph{ $\Phi(p_{1},\ldots,p_{m})$.}}\textup{}\\
\textup{\emph{Output:}}\textup{ A CAF $F(x_{1},\ldots,x_{n})$
such that\begin{equation}
F(x_{1},\ldots,x_{n})\Longleftrightarrow\Phi(F_{1}(x_{1},\ldots,x_{n}),\ldots,F_{m}(x_{1},\ldots,x_{n}))\label{equiv3}\end{equation}
}
\begin{enumerate}
\item Let $r_{1}<\ldots<r_{l}$ be all real roots of $(P_{F_{1}}\cup\ldots\cup P_{F_{m}})\cap\mathbb{R}[x_{1}]$.
\item For $1\leq i\leq l,$ set $a_{2i}(x_{1}):=(x_{1}=r_{i})$ and for
$0\leq i\leq l,$ set $a_{2i+1}(x_{1}):=(r_{i}<x_{1}<r_{i+1})$, where
$r_{0}:=-\infty$ and $r_{l+1}:=\infty$.
\item For $1\leq i\leq2l+1$

\begin{enumerate}
\item For $1\leq j\leq m$, let $G_{j}$ be the formula $G$ found using
Lemma \ref{lem1} applied to $F_{j}$ and $a_{i}$.
\item Let $j_{1},\ldots,j_{s}$ be all $1\leq j\leq m$ for which $G_{j}$
is neither $true$ nor $false$.
\item Let $\Psi(p_{j_{1}},\ldots p_{j_{s}})$ be the formula obtained from
$\Phi$ by replacing $p_{j}$ with $G_{j}$ for all $j$ for which
$G_{j}$ is $true$ or $false$.
\item If $\Psi$ is $true$ or $false$, set $H_{i}(x_{1},\ldots,x_{n}):=a_{i}\wedge\Psi$.
\item Otherwise set $H_{i}(x_{1},\ldots,x_{n})$ to \[
CombineStacks(a_{i};G_{j_{1}},\ldots,G_{j_{s}};\Psi)\]

\end{enumerate}
\item Return $F(x_{1},\ldots,x_{n}):=\bigvee_{1\leq i\leq2l+1}H_{i}(x_{1},\ldots,x_{n})$.
\end{enumerate}
\end{algorithm}
\begin{proof}
(Correctness of the algorithms) To show correctness of \emph{CombineStacks},
let us first show that inputs to \emph{Lift} satisfy the required
conditions. Condition (1) follows from step 1 of \emph{CombineStacks}.
If $k=2$, the cell $C$ is defined as a root or the open interval
between two subsequent roots of polynomials of $W_{1}$. For $k>2$,
the cell $C$ is defined as a graph of a root or the set between graphs
of two subsequent roots of polynomials of $W_{k-1}$ over a cell on
which $W_{k-1}$ is delineable. This proves condition (2). Conditions
(3) and (4) are guaranteed by Lemmas \ref{lem1} and \ref{lem2}.
Finally, (5) is satisfied, because $\Phi$ is always the same formula,
given as input to \emph{CombineStacks}. 

To complete the proof we need to show the equivalences (\ref{equiv1})
and (\ref{equiv2}). Equivalence (\ref{equiv2}) follows from Lemma
\ref{lem2} and the fact that the sets \[
\{(x_{1},\ldots,x_{k}):(x_{1},\ldots,x_{k-1})\in C\wedge a_{i}(x_{1},\ldots,x_{k})\}\]
are disjoint and their union is equal to $C\times\mathbb{R}$. Equivalence
(\ref{equiv1}) follows from Lemma \ref{lem1} and the fact that the
sets $\{x_{1}\in\mathbb{R}:a_{i}(x_{1})\}$ are disjoint and their
union is equal to $\mathbb{R}$. 

Correctness of \emph{CAFCombine} follows from Lemma \ref{lem1}, correctness
of \emph{CombineStacks,} and the fact that the sets $\{x_{1}\in\mathbb{R}:a_{i}(x_{1})\}$
are disjoint and their union is equal to $\mathbb{R}$. \end{proof}
\begin{example}
Let\begin{eqnarray*}
f_{1} & := & (x+1)^{4}+y^{4}-4\\
g_{1} & := & (x+2)^{2}+y^{2}-5\\
f_{2} & := & (x-1)^{4}+y^{4}-4\\
g_{2} & := & (x-2)^{2}+y^{2}-5\end{eqnarray*}
and let \begin{eqnarray*}
A_{1} & := & \{(x,y)\in\mathbb{R}^{2}\::\: f_{1}<0\wedge g_{1}<0\}\\
A_{2} & := & \{(x,y)\in\mathbb{R}^{2}\::\: f_{2}<0\wedge g_{2}<0\}\end{eqnarray*}
The following CAFs represent cell decompositions of $A_{1}$ and $A_{2}$.\begin{eqnarray*}
F_{1}(x,y) & := & r_{1}<x<r_{2}\wedge Root_{y,1}f_{1}<y<Root_{y,2}f_{1}\vee\\
 &  & x=r_{2}\wedge Root_{y,1}f_{1}<y<Root_{y,2}f_{1}\vee\\
 &  & r_{2}<x<r_{4}\wedge Root_{y,1}g_{1}<y<Root_{y,2}g_{1}\\
F_{2}(x,y) & := & r_{3}<x<r_{5}\wedge Root_{y,1}g_{2}<y<Root_{y,2}g_{2}\vee\\
 &  & x=r_{5}\wedge Root_{y,1}g_{2}<y<Root_{y,2}g_{2}\vee\\
 &  & r_{5}<x<r_{6}\wedge Root_{y,1}f_{2}<y<Root_{y,2}f_{2}\end{eqnarray*}
where\begin{eqnarray*}
r_{1} & := & -1-\sqrt{2}\approx-2.414\\
r_{2} & := & Root_{x,1}x^{4}+6x^{3}+10x^{2}-2x-1\approx-0.244\\
r_{3} & := & 2-\sqrt{5}\approx-0.236\\
r_{4} & := & -2+\sqrt{5}\approx0.236\\
r_{5} & := & Root_{x,2}x^{4}-6x^{3}+10x^{2}+2x-1\approx0.244\\
r_{6} & := & 1+\sqrt{2}\approx2.414\end{eqnarray*}
Compute a CAF representation of $A_{1}\cap A_{2}$ (Figure 1) using
\emph{CAFCombine}. 

The input consists of $F_{1}$, $F_{2}$ and $\Phi(p_{1},p_{2}):=p_{1}\wedge p_{2}$.
The roots computed in step (1) are $r_{1}$, $r_{2}$, $r_{3}$, $r_{4}$,
$r_{5}$ and $r_{6}$. In step (3), for all $i\neq7$, either $G_{1}$
or $G_{2}$ is $false$, and hence $\Psi=false$. For $i=7$ the algorithm
computes \emph{$CombineStacks(a_{7};G_{1},G_{2};\Phi)$,} where \begin{eqnarray*}
a_{7} & := & r_{3}<x<r_{4}\\
G_{1} & := & Root_{y,1}g_{1}<y<Root_{y,2}g_{1}\\
G_{2} & := & Root_{y,1}g_{2}<y<Root_{y,2}g_{2}\end{eqnarray*}
The projection sequence computed in step (1) of \emph{CombineStacks}
is \begin{eqnarray*}
W_{2} & := & \{g_{1},g_{2}\}\\
W_{1} & := & \{x,x^{2}+4x-1,x^{2}-4x-1\}\end{eqnarray*}
The only root of $W_{1}$ in $(r_{3},r_{4})$ is $0$. The returned
cell decomposition of $A_{1}\cap A_{2}$ consists of three cells constructed
by \emph{Lift} over cells $r_{3}<x<0$, $x=0$, and $0<x<r_{4}$.
\begin{eqnarray*}
F(x,y) & := & r_{3}<x<0\wedge Root_{y,1}g_{2}<y<Root_{y,2}g_{2}\vee\\
 &  & x=0\wedge-1<y<1\vee\\
 &  & 0<x<r_{4}\wedge Root_{y,1}g_{1}<y<Root_{y,2}g_{1}\end{eqnarray*}
Note that the computation did not require including $f_{1}$ and $f_{2}$
in the projection set.
\end{example}
\begin{figure}
\caption{Sets $A_{1}$ and $A_{2}$}

\includegraphics[scale=0.64]{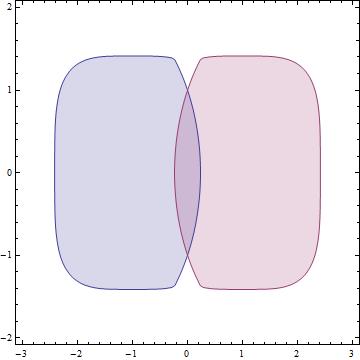}

\end{figure}

\section{\label{SubdivideSec}Algorithm Subdivide}

In this section we propose an algorithm for subdividing polynomial
systems $S(x_{1},\ldots,x_{n})$ given in disjunctive normal form.
From experimenting with various subdivision methods we deduced the
following rules for designing a subdivision heuristic.
\begin{itemize}
\item Do not subdivide conjunctions.
\item Group terms of the disjunction so that different groups have as few
common polynomials as possible.
\item Common polynomials that contain $x_{n}$ matter much more than polynomials
that do not contain $x_{n}$.
\item Common polynomials with higher degrees in $x_{n}$ matter more than
polynomials with lower degrees in $x_{n}$.
\end{itemize}
This led us to the following subdivision algorithm. The algorithm
depends on a parameter $0\leq p\leq1$ to be determined experimentally.
\begin{notation}
Let $S(x_{1},\ldots,x_{n})$ and $T(x_{1},\ldots,x_{n})$ be real
polynomial systems. Let $Wt(S)$ denote the sum of degrees in $x_{n}$
of all distinct polynomials that appear in $S$, and let $Wt(S,T)$
denote the sum of degrees in $x_{n}$ of all distinct polynomials
that appear both in $S$ and in $T$.\end{notation}
\begin{algorithm}
\label{AlgSubdivide}(Subdivide)\\
Input: \emph{A real polynomial system} $S(x_{1},\ldots,x_{n})$.\textup{}\\
\textup{\emph{Output:}}\textup{ A }Boolean formula $\Phi$ and
real polynomial systems $P_{1},\ldots,P_{m}$ such that \[
S(x_{1},\ldots,x_{n})\Leftrightarrow\Phi(P_{1}(x_{1},\ldots,x_{n}),\ldots,P_{m}(x_{1},\ldots,x_{n}))\]

\begin{enumerate}
\item If $S$ is not a disjunction return $\Phi:=Id$ and $P_{1}:=S$.
\item Let $S=S_{1}\vee\ldots\vee S_{k}$. Construct a graph $G$ as follows.

\begin{enumerate}
\item The vertices of $G$ are $S_{1},\ldots,S_{k}$.
\item There is an edge connecting $S_{i}$ and $S_{j}$ if $Wt(S_{i},S_{j})>0$
and \[
Wt(S_{i},S_{j})\geq p\min(Wt(S_{i}),Wt(S_{j}))\]

\end{enumerate}
\item Compute the connected components $\{S_{1,1},\ldots,S_{1,l_{1}}\},\ldots,\{S_{m,1},\ldots,S_{m,l_{m}}\}$
of $G$.
\item For $1\leq j\leq m$ set $P_{j}:=S_{j,1}\vee\ldots\vee S_{j,l_{j}}$.
\item Set $\Phi(p_{1},\ldots,p_{m}):=p_{1}\vee\ldots\vee p_{m}$
\item Return $\Phi$ and $P_{1},\ldots,P_{m}$.
\end{enumerate}
\end{algorithm}

\section{Empirical Results}

In this section we compare performance of\emph{ DivideAndConquerCAD}
and of direct CAD computation. As benchmark problems we chose formulas
obtained by application of virtual term substitution methods to quantifier
elimination problems, because such formulas are {}``naturally occurring''
CAD inputs that are disjunctions of many terms. We ran four variants
of \emph{DivideAndConquerCAD,} corresponding to different choices
of the parameter $p$ in \emph{Subdivide.} We used $p=0.25$, $p=0.5$,
$p=0.75$, and $p=1$. The algorithms have been implemented in C,
as a part of the kernel of \emph{Mathematica}. For direct CAD computation
the algorithms use the \emph{Mathematica} implementation of the version
of CAD described in \cite{S7}. The experiments have been conducted
on a Linux virtual machine with a $3.07$ GHz Intel Core i7 processor
and $6$ GB of RAM available. Each computation was given a time limit
of $1200$ seconds.

\subsection{Examples from \cite{BS}}

In \cite{BS} there are 56 examples obtained by application of virtual
term substitution methods to quantifier elimination problems. Here
we used 28 of the examples, 7 from applications and 21 randomly generated,
for which there were between 2 and 4 free variables. In 22 of the
28 examples at least one method finished within the time limit and
the difference between the slowest and the fastest timing was at least
$10\%$. In 10 of the 22 examples the input systems were subdivided
only for $p=1$, and \emph{DivideAndConquerCAD} with $p=1$ was slower
than direct CAD computation. Timings for the remaining 12 examples
are shown in Figure \ref{Fig1} (note the logarithmic scale). In 9
of the 12 examples \emph{DivideAndConquerCAD} with $p=0.75$ is faster
than direct CAD computation, in 3 examples it is slower. In the 12
examples \emph{DivideAndConquerCAD} with $p=0.75$ is the fastest
method on average, $2.25$ times faster than direct CAD computation.

\begin{figure}

\caption{\label{Fig1}Examples from \cite{BS}}
\includegraphics[width=0.95\textwidth]{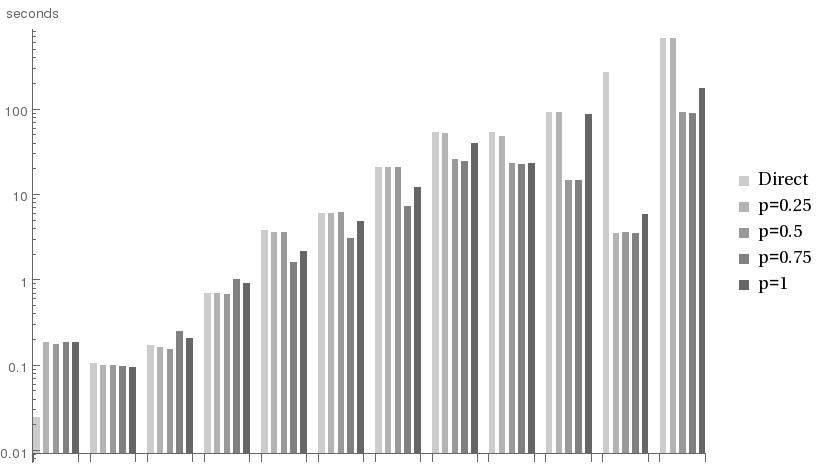}

\end{figure}

\subsection{Random examples}

We generated 16 random examples with 2 or 3 variables. The examples
were obtained by elimination of up to three quantifiers using virtual
term substitution (with intermediate formula simplification). The
initial quantified systems were randomly generated quantified conjunctions
of 2-4 polynomial equations or inequalities. The polynomials were
linear in all quantified variables except for the first one and quadratic
in the remaining variables. The quadratic term in the first quantifier
variable did not contain other quantifier variables. The results of
quantifier elimination were put in disjunctive normal form and only
disjunctions of at least 10 terms were selected. In 2 examples all
timings were the same. Timings for the remaining 14 examples are shown
in Figure \ref{Fig2} (note the logarithmic scale). In 12 of the 14
examples \emph{DivideAndConquerCAD} with $p=0.75$ is faster than
direct CAD computation, in one example it is slower. In the 14 examples
\emph{DivideAndConquerCAD} with $p=0.75$ is the fastest method on
average, at least $3.96$ times faster than direct CAD computation
(in two examples direct CAD computation did not finish in $1200$
seconds and \emph{DivideAndConquerCAD} with $p=0.75$ did).

\begin{figure}

\caption{\label{Fig2}Random examples}
\includegraphics[width=0.95\textwidth]{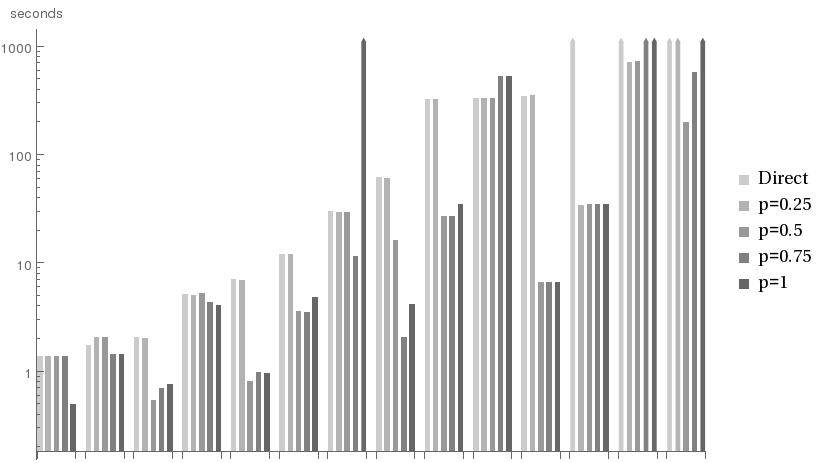}

\end{figure}

\subsection{Conclusions}

The experiments show that \emph{DivideAndConquerCAD} with graph-based
\emph{Subdivide} is often, but not always, faster than direct CAD
computation\emph{.} On average, the best performance was obtained
by choosing the parameter value $p=0.75$ in \emph{Subdivide}.

\end{document}